\documentclass{article}
\usepackage{mathtools}
\usepackage{amsmath,amssymb}

\usepackage[margin=1.5in]{geometry}

\usepackage{times}
\usepackage{soul}
\usepackage{url}
\usepackage[hidelinks]{hyperref}
\usepackage[utf8]{inputenc}
\usepackage[small]{caption}
\usepackage{graphicx}
\usepackage{amsmath, amsfonts}
\usepackage{amsthm}
\usepackage{booktabs}
\usepackage{algorithm}
\usepackage{algorithmic}
\usepackage{comment}
\DeclareMathOperator{\EX}{\mathbb{E}}

\newtheorem{theorem}{Theorem}

\newtheorem{corollary}{Corollary}

\newtheorem{lemma}{Lemma}
\newtheorem{definition}{Definition}

\title{Concentration of Distortion: \\The Value of Extra Voters in Randomized Social Choice}

\author{Brandon Fain$^1$, William Fan$^2$, Kamesh Munagala$^1$ \\ $^1$ Department of Computer Science, Duke University, USA \\ $^2$ North Carolina School of Science and Mathematics, USA}

\begin{document}

\maketitle

\begin{abstract}
We study higher statistical moments of Distortion for randomized social choice in a metric implicit utilitarian model. The Distortion of a social choice mechanism is the expected approximation factor with respect to the optimal utilitarian social cost (OPT). The $k^{th}$ moment of Distortion is the expected approximation factor with respect to the $k^{th}$ power of OPT. We consider mechanisms that elicit alternatives by randomly sampling voters for their favorite alternative. We design two families of mechanisms that provide constant (with respect to the number of voters and alternatives) $k^{th}$ moment of Distortion using just $k$ samples if all voters can then participate in a vote among the proposed alternatives, or $2k-1$ samples if only the sampled voters can participate. We also show that these numbers of samples are tight. Such mechanisms deviate from a constant approximation to OPT with probability that drops exponentially in the number of samples, independent of the total number of voters and alternatives. We conclude with simulations on real-world  Participatory Budgeting data to qualitatively complement our theoretical insights. 
\end{abstract}

\section{Introduction}
\label{sec:intro}
For many problems in social choice, the number of alternatives is very large. For example, consider the problem of voting over possible budgets in a given municipality, where the number of alternatives is infinite (for a divisible budget) or exponential (for funding integral projects). In such settings, it may be impractical to elicit full rankings over alternatives from every voter. Instead, we may want to design mechanisms that only require voters to rank at most a constant number of alternatives. In this paper, we study such mechanisms.

We consider the standard problem in social choice wherein there is a set $N$ of $n$ voters and a set $M$ of alternatives from which we must select a single winner. However, we assume that $|M|$ is large enough to prohibit eliciting full rankings over the alternatives. We also allow $n$ to be large. We adopt the implicit utilitarian perspective with metric constraints~\cite{implicitUtilitarian,YuCheng,metricSocialChoiceRandomDictator,metricSocialChoiceLowerBounds,deterministic,Feldman}. That is, we assume that voters have cardinal costs over alternatives, and these costs are constrained to be metric, but voters cannot directly report cardinal costs.  We want to design social choice mechanisms to minimize the total social cost by only asking voters to rank at most a constant number of alternatives. We measure the efficiency of a mechanism as its \textit{Distortion} (see Section~\ref{sec:prelim}), the worst case approximation to the total social cost. 


It is easy to see that randomization is necessary to achieve constant Distortion if we cannot elicit the ordinal preferences of voters over all alternatives. One natural form of randomization is to elicit alternatives by randomly sampling voters and querying them for their favorite alternatives. More generally, in this paper we mechanisms of the following type: The set of alternatives will be the favorite alternatives expressed by a subset of the voters. Subsequently, these alternatives are ranked either (i) by the entire population of the voters or (ii) by a small subset of the voters. We refer to these as the full and limited participation models respectively.

These assumptions are not merely of theoretical interest, but model social choice in emergent domains. Assumption (i) is natural in contexts where all voters are entitled to participate in the final election. For instance, in real-world Participatory Budgeting applications (see Section~\ref{sec:empirical}), a small subset of individuals propose projects, but a much larger number participate in the subsequent vote. Assumption (ii) models situations where we want a lightweight social choice mechanism that only involves a small number of voters overall such as the many department level decisions made at universities by committees representing samples of the faculty.  

Prior work~\cite{metricSocialChoiceRandomDictator,2-agree,randomReferee} analyzed simple social choice mechanisms for achieving constant Distortion. However, focusing on the \textit{expected} Distortion can yield randomized mechanisms that can deviate significantly from their expectation ex-post, and hence may be risky to implement in practice. 

We address this problem by considering higher moments of Distortion. The $k^{th}$ moment of Distortion is the expected approximation factor with respect to the $k^{th}$ power of the optimal utilitarian social cost. The goal of bounding higher moments of Distortion is directly analogous to providing high probability bounds on approximation guarantees with respect to the total social cost. We note that obtaining such a bound does not follow in a trivial manner from standard sampling arguments: The higher moments depend on the entire distribution of the Distortion obtained by the mechanism, and if this distribution has unbounded variance, then it is not possible to bound the second moment  by a constant with any number of samples, let alone higher moments. Moreover, it is initially unclear how to take the ``best'' result out of many randomly sampled alternatives. Our key insight is that the metric assumption enables us to derive tight bounds on higher moments with only a few samples by using existing deterministic social choice rules to take the ``best'' from many randomly sampled alternatives.

\subsection{Summary of Results} 
Our primary contribution is the development and analysis of randomized social choice mechanisms that achieve constant $k^{th}$ moment of Distortion in the metric implicit utilitarian model while requiring each voter to rank at most $O(k)$ sampled alternatives, regardless of the total number of voters and alternatives. The normalized $k^{th}$ moment of Distortion is defined formally in Section~\ref{sec:prelim}, and our results are summarized in Table~\ref{table:sample}.  In particular, we design two families of mechanisms that have constant $k^{th}$ moment of Distortion. The first asks just $k$ randomly chosen voters for their favorite alternatives, assuming all $n$ voters can subsequently participate in a vote among these alternatives. The second asks $2k-1$ voters for their favorite alternatives, and only these sampled voters participate in a vote among their favorite alternatives. To the best of our knowledge, these are the first results in implicit utilitarian social choice providing guarantees for arbitrarily high moments of Distortion and approximating the optimal social cost with high probability.

\begin{table}[htbp]
\centering
\begin{tabular}{ || c | c | c || } 
 \hline
 Participation Model & Lower bound & Upper bound \\ \hline
 Full & $k$ (Thm.~\ref{thm:fullLB}) & $k$ (Thm.~\ref{thm:fullUB}) \\ \hline
 Limited & $2k-1$ (Thm.~\ref{thm:limLB}) & $2k-1$ (Thm.~\ref{thm:limUB}) \\ \hline
\end{tabular}
\caption{The number of samples of favorite alternatives of voters for achieving constant normalized $k^{th}$ moment of Distortion.} 
\label{table:sample}
\end{table}

Additionally, we show that our upper bounds on the number of samples needed are tight.  We show that the $k^{th}$ moment of Distortion is unbounded in the following two settings: First, when we only sample $k-1$ favorite alternatives and all $n$ voters can subsequently compare these alternatives, and secondly, when we only sample $2k-2$ voters and the entire mechanism uses only their favorite alternatives and their comparisons between these alternatives. From a practical perspective, we demonstrate the value of using additional voters and alternatives: At most two additional samples guarantee that another higher moment of Distortion can be bounded. Finally, in Section~\ref{sec:empirical}, we present simulations on real-world  Participatory Budgeting data to qualitatively complement our theoretical insights.

\subsection{Related Work}
\label{sec:related}
\subsubsection{Metric Distortion.} The Distortion of randomized social choice mechanisms in metrics is well studied~\cite{implicitUtilitarian,metricSocialChoiceRandomDictator,metricSocialChoiceLowerBounds,2-agree}. The Random Dictatorship mechanism samples the favorite alternative of a single voter, and the 2-Agree mechanism~\cite{2-agree} samples at most $\min(n+1, m+1)$ favorite alternatives of voters. Random Dictatorship has Distortion at most 3~\cite{metricSocialChoiceRandomDictator}, and 2-Agree improves this when $m$ is small. Nothing better than Random Dictatorship is known if the goal is to minimize the Distortion. However, it is easy to show that such mechanisms do not have constant second (or higher) moment of Distortion~\cite{randomReferee}.
 
Using the second moment of Distortion as a proxy for risk was introduced in~\cite{sequentialDeliberation,randomReferee}, where it was shown that making one sampled voter compare the favorite alternatives of two randomly sampled voters bounds the second moment of Distortion. In this paper, we consider the natural question: {\em What is the value of each additional voter in how well the Distortion concentrates?} We provide a tight characterization by bounding not just the second moment, but any higher moment of Distortion.

The extreme case where $k = n$ is the deterministic setting, where it is known that the Copeland mechanism, or any mechanism based on choosing from the uncovered set~\cite{uncovered}, yields Distortion of 5~\cite{deterministic}. This bound was improved to $4.236$ in~\cite{MunagalaW19} via a weighted generalization of the uncovered set. However, both of these methods require eliciting full ordinal preferences from voters.
 
\subsubsection{Communication and Sample Complexity.} For a more thorough survey on the complexity of eliciting ordinal preferences to implement social choice rules, we refer the interested reader to~\cite{computationalSocialChoice}. ~\cite{communicationComplexity} comprehensively characterizes the \textit{communication complexity} (in terms of the number of bits communicated) of common deterministic voting rules.~\cite{sequentialElimination} and~\cite{votingCommunicationComplexity} design social choice mechanisms with low communication complexity when there are a small number of voters, but potentially a large number of alternatives. 

~\cite{winnerPrediction,marginPrediction} study the sample complexity of predicting the outcome of deterministic social choice rules. However, a ``sample'' in this work is the entire ordinal preference list for a single voter, whereas a sample for us is only the top alternative for a given voter. Even then, they show that predicting the outcome of rules with small Distortion (such as Copeland) requires a number of samples that grows with the total number of alternatives. We show that a smaller number of more limited samples suffice to bound higher moments of Distortion.
 
 Recently,~\cite{Nisarg1} studied a different notion of communication complexity in a non-metric implicit utilitarian model where voters can communicate bits of information about their {\em cardinal} preferences. In this case, the baseline is ordinal voting, and the other extreme is communicating the entire set of cardinal utilities. They show tight results for how Distortion trades off with the communication complexity in terms of bits of information communicated per voter. In our setting, voters only convey ordinal information and we study the {\em sample complexity} to bound not just the Distortion but also how well it concentrates.
 

\section{Preliminaries}
\label{sec:prelim}
We have a set $N$ of $n$ voters and a set $M$ of alternatives, from which we must choose a single outcome. For each agent $i \in N$ and alternative $a \in M$, there is some underlying dis-utility $d(i,a) \geq 0$. Let $p_i = \mbox{argmin}_{a \in M} d(i,a)$, that is, $p_i$ is the favorite alternative for voter $i$. Ordinal preferences are specified by a total order $\sigma_i$ consistent with these dis-utilities (i.e., an alternative is ranked above another only if it has lower dis-utility). A preference profile $\sigma$ specifies the ordinal preferences of all agents, and we denote $\sigma \in \rho(d)$ to mean that $\sigma_i$ is consistent with the dis-utilities for every $i$. A deterministic social choice rule is a function $f$ that maps a preference profile $\sigma$ to an alternative $a \in M$. A randomized social choice rule maps a preference profile $\sigma$ to a distribution over $M$. 

\subsection{Metric Implicit Utilitarian Model} We measure the quality of an alternative $a \in S$ by its {\em social cost}, given by $ SC(a,d) = \frac{1}{n} \sum_{i \in N} d(i,a)$. Where $d$ is obvious from context, we will simply write $SC(a)$. Let $a^* \in M$ be the minimizer of social cost. The Distortion~\cite{distortion} measures the worst case approximation to the optimal social cost of a given mechanism, in expectation for randomized mechanisms.  

\begin{definition}
The \textbf{Distortion} of a social choice rule $f$ is $$ \mbox{Distortion}(f) = \sup_{d, \, \sigma \in \rho(d)} \frac{\mathbb{E}_{f(\sigma)} [SC(a,d)]}{SC(a^*,d)}.$$ 
\end{definition}

We assume that $M \cup N$ is a set of points in a metric space. Specifically, we assume the disutility function $d$ is the distance function over this metric space. This assumption models social choice scenarios where there is an objective notion of the distance between alternatives. The metric assumption is common in the implicit utilitarian literature~\cite{metricSocialChoiceRandomDictator,metricSocialChoiceLowerBounds,sequentialDeliberation,2-agree,YuCheng,deterministic,Feldman,randomReferee}, and we consider an example from participatory budgeting in Section~\ref{sec:empirical} where the metric assumption is plausible. 

\subsection{Sampling and Higher Moments of Distortion} We consider mechanisms that implement a randomized social choice rule by first eliciting favorite alternatives from a random sample of voters and then uses only these alternatives for the rest of the mechanism. The size of this random sample is the {\em sample complexity} of our mechanism. We are interested in mechanisms with constant sample complexity with respect to $n$ and $m$. A mechanism with sample complexity $s$ only requires voters to rank at most $s$ alternatives, so constant sample complexities implies that the number of alternatives voters must rank is constant with respect to $n$ and $m$.   

We consider two models that differ in how voters participate after we elicit these alternatives. In the \textit{full participation model} of Section~\ref{sec:sample} we allow all voters to rank the alternatives from the first step and we aggregate these votes to output the winner. While this requires two distinct rounds, it is close to how real Participatory Budgeting processes work, where proposals are constructed by a subset of the population in the first stage, and these are put to vote in the second stage. In the \textit{limited participation model} of Section~\ref{sec:query}, only the sample of voters from the first step vote over the alternatives. Thus, mechanisms in the limited participation model do not require a second distinct round involving different voters. It is worth noting that while the sample complexity of our results are lower in the full participation model, the total communication complexity is higher because all voters participate in the second round.



In order to capture the notion of risk inherent in a randomized social choice mechanism, we consider higher statistical moments of Distortion. In order to fairly compare the bounds for different moments, we normalize by the $k^{th}$ root.

\begin{definition}
The \textbf{normalized} $\mathbf{k^{th}}$\textbf{ moment of Distortion} of a social choice rule $f$ is $$ \mbox{Distortion}^k(f) = \sup_{d, \, \sigma \in \rho(d)} \frac{ \left( \mathbb{E}_{f(\sigma)} \left[ (SC(a,d))^k \right]\right)^{1/k}}{SC(a^*,d)}.$$ 
\end{definition}

Note that by Jensen's inequality, if a mechanism $f$ has $Distortion^k(f) \leq c$ then $Distortion^{k'}(f) \leq c$ for all $k' \leq k$. By contrast, lower moments do \textit{not} imply anything about higher moments of Distortion. 

\subsection{Relationship Between Higher Moments and High Probability Guarantees}
Upper bounds on higher moments of Distortion immediately provide high probability guarantees for approximating the optimal social cost via Markov's inequality (see Corollaries~\ref{cor:PRCProb} and~\ref{cor:FRCProb}). However, one can reasonably ask whether the high probability bounds we achieve in this way are ``tight.'' 

More precisely, suppose we want to approximate the optimal social cost with high probability: i.e., for constant $c > 1$, find an alternative $a$ such that $SC(a, d) \leq c \cdot SC(a^*, d)$ with probability at least $1-\delta$. How many samples (favorite alternatives of random voters) are necessary as a function of $c$ and $\delta$? The example in Theorem 1 shows that one needs at least $\frac{\log(1/\delta)}{\log(c+1)}$ samples in the full participation model. On the other hand, Corollary 2 shows that our PRC mechanism needs just $\frac{\log(1/\delta)}{\log(c/11)}$ samples (for $c > 11$). So our results are tight with respect to the dependence on the probability term $\delta$, but the factor of 11 in Corollary 1 is a consequence of the analysis for Theorem 2 and may be improvable. 


\section{Full Participation Model}
\label{sec:sample}
In this section, we consider mechanisms that first elicit alternatives by sampling a number of voters and querying them for their most preferred alternatives and then apply a social choice rule on the elicited alternatives with all voters.  We begin with the lower bound on the number of samples needed to bound the $k^{th}$ moment of Distortion.

\begin{theorem}
\label{thm:fullLB}
Any mechanism $f$ with sample complexity less than $k$ has $Distortion^k(f) = \Omega(n^{1/k})$. 
\end{theorem}
\begin{proof}
Consider a metric space with two outcomes $A$ and $B$ separated by distance $1$. The fraction of voters located at $A$ is $\alpha > 1/2$ and at $B$ is $1-\alpha$. Note that the average (per-voter) social cost of $OPT$ is $1-\alpha$. If $k-1$ voters are sampled, with probability $(1-\alpha)^{k-1}$, all of them lie at $B$, in which case any voting mechanism using these samples is run on only outcome $B$. Therefore, the social cost in this case is $\alpha$. The $k^{th}$ moment of Distortion is therefore at least:
$$ \left( (1-\alpha)^{k-1} \left( \frac{\alpha^k}{(1-\alpha)^k} \right) \right)^{1/k} = \frac{\alpha}{(1-\alpha)^{1/k}}$$
Choosing $\alpha = 1-c/n$ for constant $c$ so that all but $c$ voters lie at $A$, the above expression is $ \Omega(n^{1/k})$.
\end{proof}

\subsection{The $\mathbf{PRC_s}$ Mechanism}
On the constructive side, we consider a family of mechanisms that achieve constant normalized $k^{th}$ moment of Distortion using the minimum possible number of samples. We call this family Partially Random Copeland rules. 

\begin{definition}
	The \textbf{Partially Random Copeland} rule parameterized by positive integer $s$, denoted $\mathbf{PRC_s}$, proceeds as follows.  First sample $s$ voters $\tilde{N}$ drawn independently and uniformly at random from $N$ with replacement.   All voters in $\tilde{N}$ are queried for their favorite alternative, and the union of all such alternatives is denoted $\tilde{M}$.  Finally, $PRC_s$ returns the winning alternative under the Copeland social choice rule with voters $N$ and alternatives $\tilde{M}$.
\end{definition} 

In the rest of this section, we will show the following. Intuitively, Theorem~\ref{thm:fullUB} asserts that every additional sample in the elicitation step of PRC provides a constant approximation to the next higher moment of Distortion.     

\begin{theorem}
\label{thm:fullUB}
For any $n \ge 3$ voters, $Distortion^k(PRC_{k}) \leq 11 + \frac{8}{n-2}$, which approaches $11$ as $n \rightarrow \infty$.
\end{theorem}

As a simple consequence, using Markov's inequality, this yields a high probability bound on Distortion. In particular, every additional sample in the elicitation step of PRC provides a geometric improvement in the high probability bound. 
\begin{corollary}
\label{cor:PRCProb}
As $n \rightarrow \infty$ and $c > 11$, the probability that $PRC_{k}$ outputs an alternative with social cost more than $c$ times that of the social optimum is at most $(11/c)^k$.
\end{corollary}

We first present a useful lemma bounding the $k^{th}$ moment of the minimum of $i.i.d.$ random variables.

\begin{lemma}
\label{lem:min}
Let $X_1, X_2, \ldots, X_k$ be drawn $i.i.d.$ from distribution $X$ and let $\mu = \EX[X]$. Then, 
$$ \left( \EX \left[ \min(X_1, X_2, \ldots , X_k)^k \right]  \right)^{1/k} \leq \mu$$
\end{lemma}
\begin{proof}
Let $X$ have support $d_1 \leq d_2 \leq \cdots \leq d_n$ with probabilities $\rho_1, \rho_2,\ldots,\rho_n.$ Then, $\mu = \sum_{i=1}^{n}d_i\rho_i.$ 
Let $P_i = \sum\limits_{y = i+1}^{n}\rho_y = \Pr[X > d_i]$. Note that $P_0 = 1$. Let $q_i = \rho_i d_i$  and $Q_i = \sum\limits_{j=i+1}^{n} \rho_j d_j$. Note that $Q_0 = \mu$ and  $Q_n = 0$. 
 Let $Y = \min(X_1, X_2, \ldots , X_k)$. Note that 
$$ \Pr[Y = d_i]  =  \Pr[Y \ge d_i] - \Pr[Y \ge d_{i+1}]  = (P_i+\rho_i)^k - P_i^k$$
Therefore, we have:
\begin{align*}
 \EX\big[Y^k\big] & =   \sum_{i=1}^{n}d_i^k\big[(P_i+\rho_i)^k - P_i^k\big]\\
& = \sum_{i=1}^{n}d_i^k\rho_i\bigg[\sum_{r=0}^{k-1}{k \choose r+1}\rho_i^r P_i^{k-1-r}\bigg]\\
 & = \sum_{i=1}^{n}d_i^k\rho_i\bigg[\sum_{r=0}^{k-1}{k \choose r+1}\rho_i^r \bigg(\sum_{j = i+1}^{n}\rho_j\bigg)^{k-1-r}\bigg]\\
 & \le \sum_{i=1}^{n}\big(d_i\rho_i\big) \sum_{r=0}^{k-1}{k \choose r+1}\big(d_i\rho_i\big)^r\bigg(\sum_{j=i+1}^{n} \rho_j d_j\bigg)^{k-1-r}\\
 & = \sum_{i=1}^{n}\bigg[q_i \sum_{r=0}^{k-1} {k \choose r+1} q_i^r Q_i^{k-1-r}\bigg] \\
 & = \sum_{i=1}^{n}\bigg[(Q_i+q_i)^k - Q_i^k\bigg]= Q_0^k - Q_n^k =  \mu^k
\end{align*}
\end{proof}

We now proceed to prove Theorem~\ref{thm:fullUB}. Let $a^* = \mbox{argmin}_{a \in M} SC(a)$ denote the social optimum. Let $\mu = SC(a^*) = \frac{1}{n} \sum_{i \in N} d(i,a^*)$. Suppose we sample a set $S$ of voters. For $i \in S$, let $X_i = d(i,a^*)$. Note that $\EX[X_i] = \mu$, and the $X_i$ are $i,i.d.$ random variables.

Let $m = \mbox{argmin}_{i \in S} X_i$ be the voter closest to $a^*$, and let $a_m$ denote their favorite alternative. Note $d(m, a_m) \le X_m$. 

Let $\alpha = 1 + \frac{1}{n-2}$. Consider a ball centered at $a^*$ of radius $\rho = 2 \alpha \mu$ denoted $B$. By Markov's inequality, we know that a strict majority, at least $\frac{n}{2} + 1$, of all voters lie within the ball $B$, since the average distance of a voter to $a^*$ is $\mu$.

Given $S$, suppose $\mathbf{PRC_k}$ chooses alternative $W$, and suppose  $d(W, a^*) = \beta \rho$. We will show an upper bound on $\beta$ using the random variable $X_m$.  Since a Copeland winner must be a member of the uncovered set~\cite{uncovered}, either a majority of voters prefer $W$ to $a_m$, or a majority of voters must prefer $W$ to an alternative $W'$ such that a majority of voters also prefer $W'$ to $a_m$. The first case is easier: if a majority of voters prefer $W$ to $a_m$, then there exists a voter $j \in B$ that prefers $W$ to $a_m$. This implies that $$\beta \rho = d(a^*, W) \leq d(a^*, j) + d(j, a_m) \leq 2 \rho + d(a^*, a_m).$$  
Recall that $X_m = d(m, a^*)$ and $d(m, a_m) \leq X_m$, so
$$\beta \rho \leq 2 \rho + 2 X_m .$$ 

The second case yields a worse bound, so we continue the analysis in that case without loss of generality. Let $d(W', a^*) = \beta' \rho$. Since a majority of voters prefer $W$ to $W'$, there is at least one voter $j \in B$ that prefers $W$ to $W'$, that is, $d(j, W) < d(j, W')$. By triangle inequality,
$$ d(j, W) \ge d(a^*, W) - d(j, a^*)  \geq (\beta - 1) \rho$$
$$ d(j, W') \le d(a^*, W') + d(j, a^*)  \leq (\beta' + 1) \rho$$
where the rightmost inequalities follow from the fact that $j \in B \implies d(j, a^*) \leq \rho$. Combining the above inequalities and assuming $\beta > 1$, we have $ \beta \leq \beta' + 2$. Similarly, if  a majority of voters prefer $W'$ to $a_m$, there exists some $l \in B$ such that $d(l, W') < d(l, a_m).$ Again, by triangle inequality:
$$ d(l, W') \ge d(a^*, W') - d(l, a^*)  \ge  (\beta'  - 1)  \rho$$
$$ d(l, a_m) \le d(l, a^*) + d(a^*,m) + d(m, a_m) \le \rho  + 2 X_m$$
where we used that $i \in B$ and $d(m,X_m) \le d(a^*,m) = X_m$.
Combining the above inequalities, we have  $\beta' < 2 + \frac{2 X_m}{\rho}$. Since $ \beta \leq \beta' + 2$, we have: $\beta < 4 + \frac{2 X_m}{\rho}$

Thus, we know that for $W$ to win Copeland,
$$d(W, a^*) = \beta \rho \leq 4 \rho + 2 X_m = 8 \left(1+\frac{1}{n-2} \right) \mu + 2 X_m$$
By triangle inequality, and using $SC(a^*) = \mu$, we have:
$$ SC(W) \le d(W, a^*) + SC(a^*) \le  \left(9+\frac{8}{n-2} \right) \mu + 2 X_m$$
Setting $\gamma = \left(9+\frac{8}{n-2} \right) \mu$, we have:
$$ \EX[SC(W)^k]  \leq  \EX[(\gamma + 2 X_m)^k] =\sum_{r=0}^{k}{k \choose r}\gamma^{k-r}\EX[X_m^r]2^r$$
Since $X_m$ is the minimum of $k$ i.i.d. random variables with mean $\mu$, applying Lemma~\ref{lem:min}, we have $\EX[X_m^k] \le \mu^k$.  Applying Jensen's inequality, for $r \le k$, we have
$$\EX[X_m^r] = \EX[(X_m^k)^{r/k}] \leq \EX[X_m^k]^{r/k} = (\mu^k)^{r/k} = \mu^r.$$ 
Therefore, we have
$$   \EX[SC(W)^k] \leq \sum_{r=0}^{k}{k \choose r} \gamma^{k-r} (2 \mu)^r = \left(11+\frac{8}{n-2}\right)^k\mu^k$$
Therefore , we have $Distortion^k(PRC_{k}) \leq 11 + \frac{8}{n-2}$, completing the proof of Theorem~\ref{thm:fullUB}.

\section{Limited Participation Model}
\label{sec:query}
In this section, we consider mechanisms that sample some number of voters, query the voters for their most preferred alternatives, and then hold an election on just the sample of voters. We first show that limiting participation in this way necessarily increases the sample complexity. 

\begin{theorem}
\label{thm:limLB}
Any anonymous limited participation randomized mechanism with sample complexity less than $2k-1$ has $Distortion^k(f) = \Omega(n^{1/k})$.
\end{theorem}
\begin{proof}
Consider the same instance as Theorem~\ref{thm:fullLB}. Suppose we sample $2k-2$ voters. Then the probability that we sample an equal number of voters located at $A$ and $B$ is 
$${2k-2 \choose k-1} \alpha^{k-1}(1-\alpha)^{k-1} \ge (2 \alpha)^{k-1} (1-\alpha)^{k-1} \ge (1-\alpha)^{k-1}$$
where we have assumed $\alpha \ge 1/2$. In this event, since there is no majority of voters in the sample that prefer either alternative, we assume that any anonymous mechanism outputs $B$ with probability at least $1/2$, so that the social cost is at least $\frac{\alpha}{2}$. Therefore, the $k^{th}$ moment of distortion is at least:
$$  \left(  (1-\alpha)^{k-1} \left( \frac{(\alpha/2)^k}{(1-\alpha)^k} \right) \right)^{1/k} = \frac{\alpha}{2 (1-\alpha)^{1/k}}$$
Choosing $\alpha = 1-c/n$ for constant $c$ so that all but $c$ voters lie at $A$, the above expression is $ \Omega(n^{1/k})$.
\end{proof}
	
\subsection{The $\mathbf{FRC_s}$ Mechanism}
Complementing the above impossibility, we show that sample complexity of $2k-1$ is also sufficient to achieve constant $k^{th}$ moment of Distortion. In particular, we define another family of social choice rules called Fully Random Copeland.

\begin{definition}
	The \textbf{Fully Random Copeland} rule parameterized by positive integer $s$, denoted $\mathbf{FRC_s}$ proceeds as follows. First samples $s$ voters $\tilde{N}$ drawn independently and uniformly at random from $N$ with replacement.   All voters in $\tilde{N}$ are queried for their favorite alternative, and the union of all such alternatives is denoted $\tilde{M}$. Finally, $FRC_s$ returns the winning alternative under the Copeland social choice rule with voters $\tilde{N}$ and alternatives $\tilde{M}$. 
\end{definition}

In the rest of this section, we will show the following. Intuitively, Theorem~\ref{thm:limUB} says that every additional two voters participating in FRC provide a constant approximation to the next higher moment of Distortion.

\begin{theorem}
\label{thm:limUB}
	$Distortion^k(FRC_{2k-1}) \leq 17.$
\end{theorem}

Again, as a consequence of Markov's inequality, we have the following high probability bound. In particular, every additional two voters in FRC provide a geometric improvement in the high probability bound.

\begin{corollary}
\label{cor:FRCProb}
	For $c \geq 17$, the probability that $FRC_{2k-1}$ outputs an alternative with social cost more than $c$ times that of the social optimum is at most $(17/c)^k$.
\end{corollary}   

As in Section~\ref{sec:sample}, we first present a result on bounding the $k^{th}$ moment of a function of $i.i.d.$ random variables; this time the function is the median instead of the minimum.

\begin{lemma}
\label{lem:median}
Let $X_1, X_2, \ldots, X_{2k-1}$ be drawn $i.i.d.$ from distribution $X$ and let $\mu = \EX[X]$. Let $Y$ denote the median of $X_1, X_2, \ldots, X_{2k-1}$ . Then, $ \left(\EX[Y^k]\right)^{1/k} \leq  4\mu$.
\end{lemma}
\begin{proof} 
We follow the same structure as the proof of Lemma~\ref{lem:min}. Let $X$ have support $d_1 \leq d_2 \leq \cdots \leq d_n$ with probabilities $\rho_1, \rho_2,\ldots,\rho_{n}$ and let $P_i = \sum\limits_{y = i+1}^{n}\rho_y$. If $Y = d_i$, then there exists a subset of $k$ values  from $X_1, X_2, \ldots, X_{2k-1}$ whose minimum is $d_i$. Using the expression for the probability of the minimum of $k$ values from Lemma~\ref{lem:min}, the probability of the latter event can be upper bounded as:
\begin{align*} \Pr[Y = d_i] & \leq {2k-1 \choose k} \Pr[\mbox{Min of $k$ draws from } X = d_i] \\
& = {2k-1 \choose k}  \big[(P_i+\rho_i)^k - P_i^k\big] \\
\implies  \EX[Y^k] & \leq {2k-1 \choose k}\sum_{i=1}^{n}d_i^k \big[(P_i+\rho_i)^k - P_i^k\big].
 \end{align*}
From the proof of Lemma~\ref{lem:min}, $ \sum_{i=1}^{n}d_i^k \big[(P_i+\rho_i)^k - P_i^k \big] \le \mu^k$. Therefore we have:
$$ \EX[Y^k] \leq {2k-1 \choose k} \mu^k  \leq {2k \choose k}\mu^k  \leq (1+1)^{2k}\mu^k  =(4\mu)^k$$
\end{proof}

We will also need the following straightforward property of the Copeland Rule.  
\begin{lemma}
\label{lem:copeland}
Suppose there are $2k-1$ alternatives and voters. We construct a tournament graph on the alternatives where there is a directed edge from alternative $S$ to alternative $T$  if at least $k$ voters strictly prefer $S$ to $T$. Then the Copeland rule always picks an alternative $W$ with in-degree strictly less than $k$. 
\end{lemma}

We now proceed to prove Theorem~\ref{thm:limUB}.
 As in the proof of Theorem~\ref{thm:fullUB}, let $a^*$ denote the optimal alternative, and let $SC(a^*) = \mu = \frac{1}{n} \sum_{i \in N} d(i,a^*)$. Suppose we sample a subset of voters, $S$ of size $2k-1$ For $i \in S$, let $X_i = d(i,a^*)$. Order these voters so that $X_1 \leq X_2 \leq \cdots \leq X_{2k-1}$ and let $m$ be the voter  that corresponds to the median of this sequence. Let $Y = d(m, a^*)$.  Note from Lemma~\ref{lem:median} that $\EX[Y^k] \leq  (4\mu)^k$.

Suppose the Copeland rule chooses an alternative $W$, and suppose $d(W, a^*) = \alpha Y$. We will find an upper bound for $\alpha$.  Consider the ball centered around $a^*$ with radius $Y$; call this $B$. By definition, at least $k$ agents in $S$ lie within $B$. Note that for any $j \in B \cap S$, $d(j,a_j) \le d(j,a^*) \le Y$. Therefore, for $j, l \in B \cap S$, we have 
$$ d(j, a_l) \le d(j,a^*) + d(l, a^*) + d(l, a_l) \le Y + Y + Y = 3Y$$
 Now, for $j \in B \cap S$, we have
$$ d(j, W) \ge d(a^*, W) - d(j,a^*) \ge (\alpha - 1) Y$$
If $\alpha > 4$, then combining the above two observations, we have that for all $j, l \in B$, we have $d(j,a_l) \le 3Y < d(j,W)$. This means that the set of at least $k$ voters in $B \cap S$ strictly prefer all of the favorite alternatives $\{a_l, l \in B \cap S\}$ to $W$.
From Lemma~\ref{lem:copeland}, this means that $W$ cannot be the Copeland winner.  Thus, for $W$ to win in the Copeland rule, we must have $\alpha \leq 4$ so $d(W, a^*)  \leq 4 Y$. By triangle inequality, 
$$SC(W) \leq SC(a^*) + d(a^*,W) \le \mu + 4Y$$ 
Using Jensen's inequality in a fashion similar to the proof of Theorem~\ref{thm:fullUB},  and using $\EX[Y^k] \leq  (4\mu)^k$ we have:
$$ \EX[SC(W)^k]  \leq \EX[(\mu + 4Y)^k]  \le \sum_{r=0}^{k} {k \choose r} \mu^k 16^r = 17^k\mu^k$$
so we have that $Distortion^k(FRC_{2k-1}) \leq 17$. This completes the proof of Theorem~\ref{thm:limUB}.

\section{Empirical Simulation}
\label{sec:empirical}
In this section, we augment our theoretical worst case analysis with a qualitative empirical demonstration of the concentration achieved by the PRC and FRC mechanisms on real world data. We use data from the Participatory Budgeting project; see~\cite{PBP}. In this domain, there are a number of public projects (such as new sidewalks, park renovations, etc.). Each project has a monetary cost, and we want to select a set of projects subject to not exceeding a total budget. In participatory budgeting, local community members vote directly over their preferred projects, and these votes are aggregated to decide which projects to fund. 

We consider knapsack voting data~\cite{PBP}, where each voter reports the set of projects they most prefer, subject to the total budget constraint. This makes knapsack voter data particularly useful for us: voters select their single alternative out of a very large space, the power set of projects. Because we also have information about the latent combinatorial space (specific projects selected and their costs), we can impose simplistic but natural notions of distance to allow us to simulate our mechanisms and study their performance with respect to the imposed distance. 

\paragraph{Simulation.} It is important to note that this is a simulation; actually running our mechanisms does not require specifying a notion of distance, and we do not know how these voters would have responded to ordinal queries in reality. We are treating an entire budget allocation as a single outcome and imputing preferences of voters over these outcomes. This reduces the problem to single winner election over a large space of alternatives in keeping with the theoretical model in the paper. Therefore, natural baseline mechanisms are single winner rules with small sample complexity, particularly Random Dictatorship which is the best-known mechanism with respect to the first moment of Distortion. Other mechanisms for participatory budgeting are tailored to specific models of voter preferences over the combinatorial space of projects, and do not, in general, provide constant Distortion guarantees for arbitrary metrics. 

\paragraph{Setup.} We consider two simple notions of distance: budget distance and Jaccard distance. Suppose there are $p$ public projects numbered $1, \hdots, p$ with costs $c_1, \hdots, c_p$, and there is a total budget of $B$. A \textit{feasible budget} is a set of projects $P$ such that $\sum_{i \in P} c_p \leq B$. The \textit{budget distance} between budgets $P$ and $Q$ is $1 - \frac{1}{B} \sum_{i \in P \cap Q} c_i$. The \textit{Jaccard distance} between $P$ and $Q$ is $1 - \frac{|P \cap Q|}{|P \cup Q|}$. The social cost of a given budget is the average distance to the proposed budgets of the voters. We use knapsack voting data from the Participatory Budgeting election held in Cambridge, MA, USA in 2015. There were 945 voters, 23 projects (implying $2^{23} > 8$ million possible budgets), and a total budget constraint of $\$600,000$. 

\paragraph{Results.} In Figure~\ref{fig:boxplotCambridge}, we present the box plots of the distributions of social cost of PRC and FRC alongside Random Dictatorship  (RD) when simulating using budget distance and Jaccard distance respectively. The RD mechanism samples a single most preferred alternative uniformly at random and has Distortion at most 3~\cite{metricSocialChoiceRandomDictator}, which is asymptotically the best known bound for any randomized social choice mechanism for arbitrary metrics. The examples qualitatively verify that the PRC and FRC mechanisms do provide substantial concentration in terms of the approximation to the optimal social cost. Furthermore, in practice we observe better average performance of PRC and FRC over that of RD, despite RD's theoretical optimality with respect to the first moment of Distortion. The results also show that FRC requires more samples to achieve similar performance as PRC. To summarize, even on real datasets, just a few additional samples provide substantially improved concentration.  

\begin{figure}[t]
\centering
\includegraphics[width=0.67\linewidth]{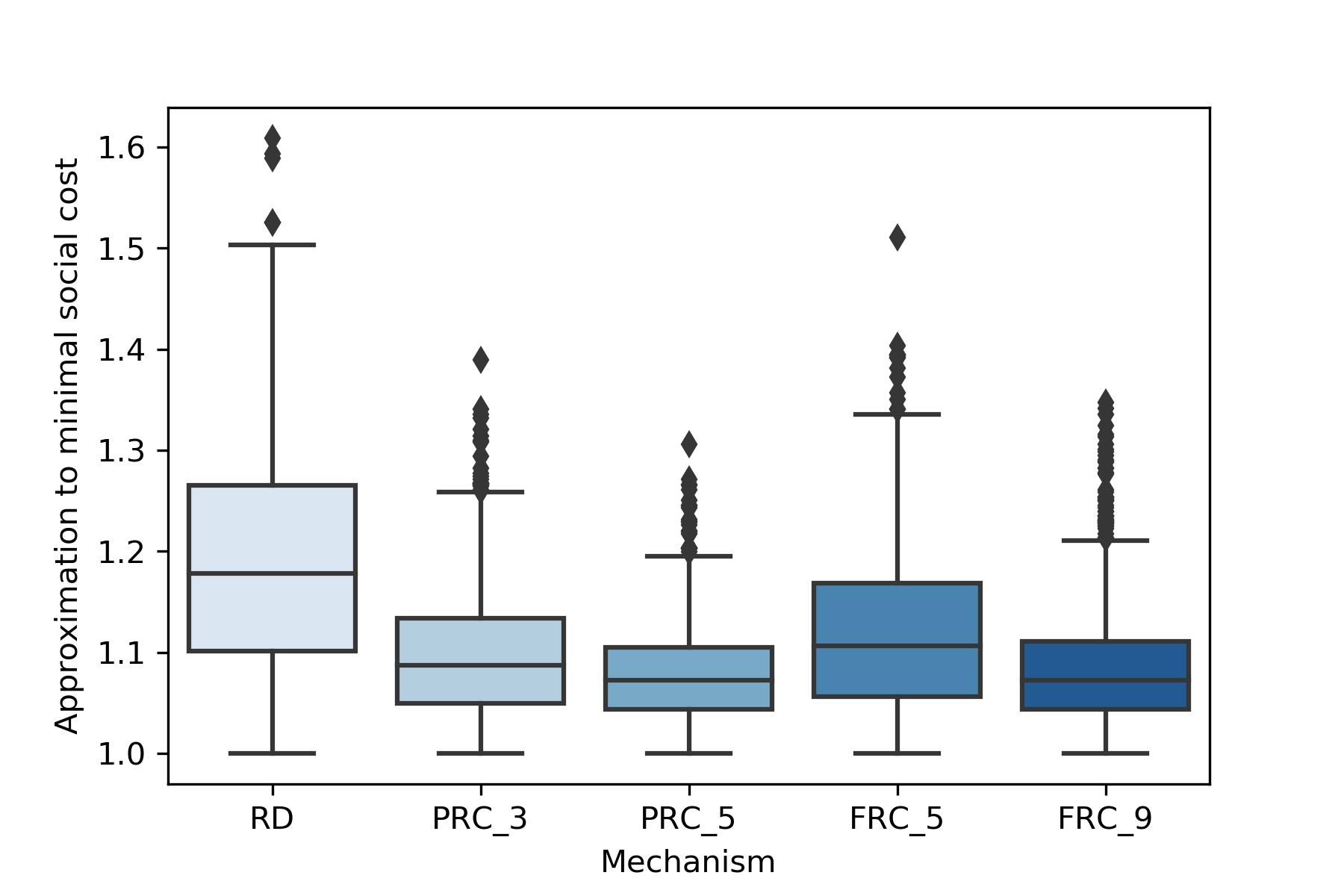} 
\includegraphics[width=0.67\linewidth]{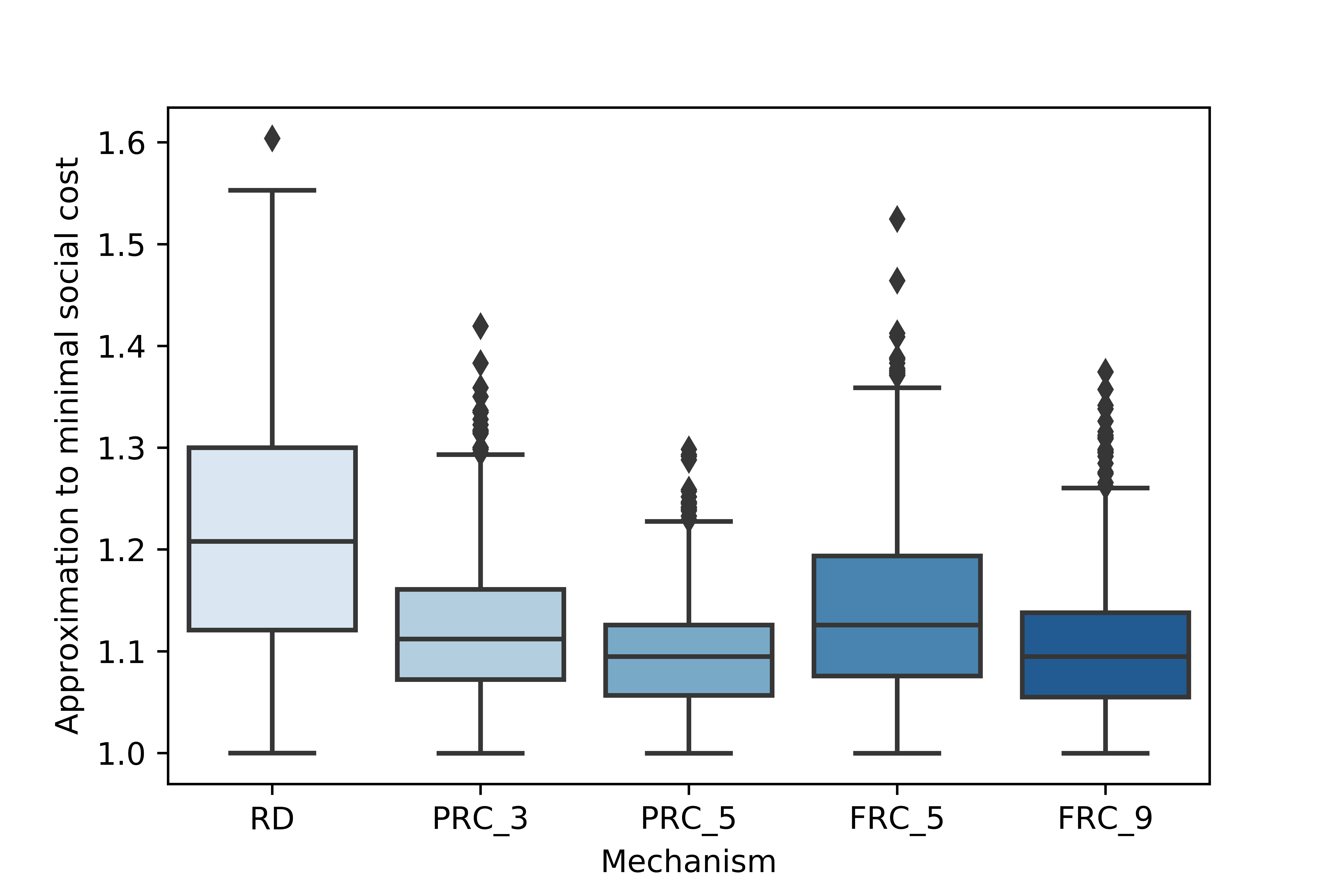}
\caption{Distribution of approximation to optimal social cost for 1,000 runs of each mechanism on Cambridge 2015 knapsack voting data using budget distance (top) and Jaccard distance (bottom).}
\label{fig:boxplotCambridge}
\label{fig:boxplotCambridgeJaccard}
\end{figure}

\section{Future Directions}
\label{sec:future}
 There are several avenues of future research. Our mechanisms involve first sampling some alternatives and then putting them to vote. Is there a truly one-shot mechanism that can bound higher moments of Distortion while only eliciting a constant amount of information from each voter with respect to the number of alternatives? Our intuition is that this should be impossible. Also, though our sample complexity bounds are tight, the exact constant in the Distortion bounds can likely be improved. However, this improvement may be nontrivial: We do not use the Distortion of Copeland as a black box, so results such as~\cite{MunagalaW19} do not directly improve our bounds. As in~\cite{Nisarg1}, it would be interesting to analyze the effect of bits of cardinal information on the sample complexity. For instance, what if we sample fewer voters, but these voters could express limited cardinal information? In a related vein, could methods that make voters interact like~\cite{sequentialDeliberation} help reduce the sample complexity of the process? 

\subsubsection{Acknowledgments} 
Kamesh Munagala was supported by NSF grants CCF-1408784, CCF-1637397, and IIS-1447554, ONR award N00014-19-1-2268, and awards from Adobe and Facebook.

\bibliographystyle{alpha}
\bibliography{refs.bib}

\end{document}